\documentclass{sig-alternate}

\conferenceinfo{ISSAC'14}{July 23--25, 2014, Kobe, Japan.}
\CopyrightYear{2014}
\crdata{978-1-4503-2501-1/14/07}
\usepackage{amsmath,amsfonts,amssymb}
\usepackage{overpic}
\usepackage{subcaption}
\usepackage{comment}
\usepackage{epic,eepic}
\usepackage{color}
\usepackage{comment}
\usepackage[all]{xy}
\usepackage{url}
\newtheorem{thm}{Theorem}

\newtheorem{lem}{Lemma}
\newtheorem{example}{Example}
\newtheorem{defn}{Definition}
\newtheorem{rem}{Remark}

\newtheorem{prop}{Proposition}

\graphicspath{{img/}}
\def \Res  {{\rm Res}}
\def \MCproj  {{\rm MCproj}}

\def \discrim  {{\rm discrim}}
\def \sqrfree  {{\rm sqrf}}

\def \lc  {{\rm lc}}
\def \Bproj  {{\tt Bp}}
\def  \zero {{\rm Zero}}
\def  \Nproj {{\tt Np}}

\def  \Hproj {{\tt Hp}}

\def \Proineq {{\tt Proineq}}

\def \RR {{\mathbb R}}
\def \ZZ {{\mathbb Z}}

\newcommand{\va}{\bm{\alpha}}
\newcommand{\vb}{\bm{\beta}}
\newcommand{\xx}{\bm{x}}
\newcommand{\yy}{\bm{y}}
\newcommand{\zz}{\bm{z}}

\def \RAGlib {{\tt RAGlib}}
\def \FI {{\tt FI}}
\def \QEPCAD {{\tt QEPCAD}}
\def \PCAD {{\tt PCAD}}
\def \TwoPro {{\tt PSD-HpTwo}}
\def \TwoHp {{\tt HpTwo}}

\def \FI {{\tt FI}}
\def \QEPCAD {{\tt QEPCAD}}
\def \PCAD {{\tt PCAD}}

\usepackage{bm}
\usepackage{algorithm,algorithmic}

\newtheorem{ex}{\qquad Example}[section]

\begin{document}
\title{Constructing Fewer Open Cells by GCD Computation in CAD Projection}
\numberofauthors{3}

\author{
 \alignauthor Jingjun Han\\
 \affaddr{School of Mathematical Sciences \& Beijing International Center for Mathematical Research\\ Peking University}
 \email{hanjingjunfdfz@gmail.com}
 \alignauthor Liyun Dai\\
 \affaddr{School of Mathematical Sciences \& Beijing International Center for Mathematical Research\\ Peking University}
 \email{dailiyun@pku.edu.cn}
 \alignauthor Bican Xia\\
 \affaddr{LMAM \& School of Mathematical Sciences\\ Peking University}
 \email{xbc@math.pku.edu.cn}
   }

\maketitle

\begin{abstract}
A new projection operator based on cylindrical algebraic decomposition (CAD) is proposed. The new operator computes the intersection of projection factor sets produced by different CAD projection orders. In other words, it computes the gcd of projection polynomials in the same variables produced by different CAD projection orders. We prove that the new operator still guarantees obtaining at least one sample point from every connected component of the highest dimension, and therefore, can be used for testing semi-definiteness of polynomials. Although the complexity of the new method is still doubly exponential, in many cases, the new operator does produce smaller projection factor sets and fewer open cells. Some examples of testing semi-definiteness of polynomials, which are difficult to be solved by existing tools, have been worked out efficiently by our program based on the new method.

\end{abstract}


\section{Introduction}
\label{secc:intro}
The cylindrical algebraic decomposition (CAD) method was first proposed by Collins \cite{collins1,Caviness1998}.
A key role in CAD algorithm is its projection operator. A well known improvement of CAD projection is Hong's projection operator \cite{Hong}. 
For many problems, a smaller projection operator given by McCallum in \cite{McCallum1,McCallum2}, with an improvement by Brown in \cite{brown}, is more efficient.

For CAD projection operators, different projection orders may lead to a great difference in complexity. Thus, in order to reduce the projection scale, it is meaningful to study the relationship between those different projection orders. For related work, see for example \cite{Dolzmann2004}.

The reason for such difference is mainly because the projection factors in the same variables produced by different projection orders may be different. For example, when we apply Brown's projection operator $\Bproj$ (see Definition \ref{def:brown projection}) to any given polynomial $f\in\ZZ[x_1,\ldots,x_n]$, it is quite possible that
$\Bproj(f,[x_{n},x_{n-1}])\neq \Bproj(f,[x_{n-1},x_{n}]).$

In this paper, we propose a new projection operator $\Hproj$ based on Brown's operator and gcd computation. The new operator computes the intersection of projection factor sets produced by different CAD projection orders. In other words, it computes the gcd of projection polynomials in the same variables produced by different CAD projection orders. In some sense, the polynomial in the projection factor sets of $\Hproj$ is irrelevant to the projection orders. We prove that the new operator still guarantees obtaining at least one sample point from every connected component of the highest dimension, and therefore, can be used for testing semi-definiteness of polynomials.

It should be mentioned that there are some non-CAD methods for computing sample points in semi-algebraic sets, such as critical point method. For related work, see for example, \cite{basu1998new,safey2003polar,el2013critical}. Mohab Safey el Din developed a Maple package \RAGlib \footnote{http://www-polsys.lip6.fr/\~{}safey/RAGLib/distrib.html} based on their work, which can test semi-definiteness of polynomials.

The new projection operator provides a possible faster way based on CAD to test semi-definiteness of polynomials in practice though the complexity of the new method is still doubly exponential. In many cases, the new projection operator produces smaller projection factor sets and thus fewer open cells than existing CAD based algorithms such as GCAD \cite{Strzebonski}. Indeed, some examples that could not be solved by existing CAD based tools have been worked out efficiently by our program based on the new operator.

The structure of this paper is as follows. In Section 2, a simple example illustrates the main idea and steps of the new projection operator \Hproj. Section 3 introduces basic definitions, lemmas and concepts of CAD. In Section 4,
the new projection operator \Hproj\ is defined and a new algorithm based on \Hproj\ is proposed. Our main result (Theorem \ref{thm:open delineable}) is proved.
In Section 5, we prove that it is valid if we replace \Bproj\ with \Hproj\ in some steps of the projection phase of the simplified CAD projection \Nproj\ we proposed recently for inequality proving in \cite{han2012}. Section 6 includes several examples which demonstrate the effectiveness of our algorithms. We conclude the paper in Section 7. 

\section{Main idea}

Let us show the comparison of our new operator and Brown's projection operator on the following simple example. Formal description and proofs of our main results are given in subsequent sections.

\begin{example}\label{ex:1}
Let $f=x^4-2x^2y^2+2x^2z^2+y^4-2y^2z^2+z^4+2x^2+2y^2-4z^2-4\in \ZZ[x,y,z].$

We first compute an open CAD (see Definition \ref{def:opencad}) defined by $f\ne 0$ in $\RR^3$ by Brown's operator.
Take the order $z\succ y \succ x$. Step 1, compute the projection polynomial $($up to a nonzero constant$)$
\[\begin{array}{rl}
f_{z} & =\Res(\sqrfree(f),\frac{\partial}{\partial z}\sqrfree(f),z)\\ \medskip
& =(x^4-2x^2y^2+y^4+2x^2+2y^2-4)(3x^2-y^2-4)^2
\end{array}\]
where ``{\em Res}" means the Sylvester resultant and ``{\em sqrf}" means ``squarefree" that is defined in Definition \ref{de:sqrfree}.

Step 2, compute the projection polynomial $($up to a nonzero constant$)$
\[\begin{array}{rl}
f_{zy} & =\Res(\sqrfree(f_z),\frac{\partial}{\partial y}\sqrfree(f_z),y)\\ \medskip
   & = (3x^2-4)(x^4+2x^2-4)(4x^2-5)^2(x-1)^8(x+1)^8
\end{array}\]
which has $8$
distinct real zeros.

Step 3, by open CAD lifting under the order $z\succ y \succ x$ and using the projection factor set $\{f_{zy}, f_z, f\}$, we will finally get $113$ sample points of $f\ne 0$ in $\RR^3$.

Now, we compute a {\em reduced open CAD} (see Algorithm \ref{de:reducedopencad}) defined by $f$ in $\RR^3$ by the new projection operator proposed in this paper.
Step 1, take the order $z\succ y \succ x$ and compute the projection polynomial $f_{zy}$ as above. Step 2, take another order $y\succ z \succ x$ and we can similarly obtain a projection polynomial $($up to a constant$)$
\[f_{yz}=(3x^2-4)^2(x^4+2x^2-4)(4x^2-5)(6x^2-7)^8.\]
Step 3, compute \[\gcd(f_{yz},f_{zy})=(3x^2-4)(x^4+2x^2-4)(4x^2-5)\] which has $6$ distinct real zeros.

Step 4, by open CAD lifting under the order $z\succ y \succ x$ and using the projection factor set $\{\gcd(f_{yz},f_{zy}), f_z, f\}$, we will finally get $87$ sample points of $f\ne 0$ in $\RR^3$.
\end{example}

\begin{rem}
The main result of this paper is Theorem \ref{thm:open delineable} which guarantees that the new projection operator can obtain at least one sample point from every connected component of the highest dimension. 

Intuitively, let us take for example three open intervals $S_1=(-a,-1)$, $S_2=(-1,1)$ and $S_3=(1,a)$ where $a=\sqrt{\sqrt{5}-1}$ and $-a,-1,1,a$ are four consecutive roots of $f_{zy}$. By typical CAD methods, $S_1, S_2, S_3$ are three open cells of $x$-axis. Note that for any open connected set $U$ of $\RR^3$ defined by $f\ne 0$ and any two points $x_1, x_2\in S_1\cup S_2\cup S_3$, we have $(\{x_1\}\times\RR^2)\cap U\ne \emptyset\Longleftrightarrow$ $(\{x_2\}\times\RR^2)\cap U\ne\emptyset$. So, by this observation, we only need to consider one open cell $(-a,a)$. For details, please see Remark \ref{rem:union}.
\end{rem}

\begin{rem}
Computing projection factor sets $($polynomials$)$ under different projection orders brings additional costs compared to traditional CAD projection operators. However, it has two gains. First, it produces fewer sample points $($representing open cells$)$ in many cases as shown in Example \ref{ex:1}. Second, the most important thing is, if the number of variables is greater than $3$, it may also reduce the scale of projection. Please see Definition \ref{def:hp}, Algorithm \ref{TwoHp}, Remark \ref{re:a1} and Remark \ref{re:a2} for details.
\end{rem}

\section{Preliminaries}
\label{sec:pre}
If not specified, for a positive integer $n$, let $\xx_n$ be the set of variable $\{x_1,\dots,x_n\}$ and $\va_n$ and $\vb_n$ denote the point $(\alpha_1,\dots,\alpha_n)\in \RR^n$ and $(\beta_1,\dots,\beta_n)\in \RR^n$, respectively.
\begin{defn}
Let $f\in \ZZ[\xx_n]$, denote by $\lc(f,x_i)$ and $\discrim(f,x_i)$ the {\em leading coefficient} and the {\em discriminant} of $f$ with respect to (w.r.t.) $x_i$, respectively. 
The set of real zeros of $f$ is denoted by $\zero(f)$. Denote by $\zero(L)$ or $\zero(f_1,\ldots,f_m)$ the common real zeros of $L=\{f_1,\ldots,f_m\}\subset \ZZ[\xx_n].$ The {\em level} for $f$ is the biggest $j$ such that $\deg(f,{x_j})>0$ where $\deg(f,x_j)$ is the degree of $f$ w.r.t. $x_j$. For polynomial set $L\subseteq \ZZ[\xx_n]$, $L^{[i]}$ is the set of polynomials in $L$ with level $i$.
\end{defn} 

\begin{defn}
	Let $P_{n}$ be the symmetric permutation group of $x_1,\ldots,x_n$. Define $P_{n,i}$ to be the subgroup of $P_{n}$, where any element $\sigma$ of $P_{n,i}$ fixes $x_1,\ldots,x_{i-1}$, i.e., $\sigma(x_j)=x_j$ for $j=1,\ldots,i-1$.
\end{defn}

\begin{defn}\label{de:sqrfree}
	If $h$$\in \ZZ[\xx_n]$ can be factorized in $\ZZ[\xx_n]$ as:
	$$h=al_{1}^{2j_1-1}\cdots l_t^{2j_t-1}{h_1}^{2i_1}\cdots {h_m}^{2i_m},$$
	where $a\in \ZZ$, $t\ge 0, m\ge 0$, $l_j$$(i=1,\ldots,t)$ and $h_i$$(i=1,\ldots,m)$ are
pairwise different irreducible primitive polynomials with positive leading coefficients $($under a suitable ordering$)$ and positive degrees in $\ZZ[\xx_n]$, then define
	\begin{align*}
        &\sqrfree(h)=l_{1}\cdots l_t{h_1}\cdots {h_m},\\ 
		&\sqrfree_1(h)=\{l_i,i=1,2,\ldots,t\},\\
		&\sqrfree_2(h)=\{h_i,i=1,2,\ldots,m\}.
	\end{align*}
	If $h$ is a constant, let $\sqrfree(h)=1,$ $\sqrfree_1(h)=\sqrfree_2(h)=\{1\}.$
\end{defn}

In the following, we introduce some basic concepts and results of CAD. The reader is referred to \cite{collins1,Hong,collins1991partial,McCallum1,McCallum2,brown} for a detailed discussion on the properties of CAD.

\begin{defn}$\cite{collins1,McCallum1}$
	An $n$-variate polynomial $f(\xx_{n-1},x_{n})$ over the reals is said to be {\em delineable} on a subset $S$ (usually connected) of $\RR^{n-1}$ if
	(1) the portion of the real variety of $f$ that lies in the cylinder $S\times \RR$ over $S$ consists of the union of the graphs of some $t\ge0$ continuous functions $\theta_1<\cdots<\theta_t$ from $S$ to $\RR$; and
	(2) there exist integers $m_1,\ldots,m_t\ge1$ s.t. for every $a\in S$, the multiplicity of the root $\theta_i(a)$ of $f(a,x_n)$ (considered as a polynomial in $x_n$ alone) is $m_i$.
\end{defn}

\begin{defn}$\cite{collins1,McCallum1}$
	In the above definition, the $\theta_i$ are called the real root functions of $f$ on $S$, the graphs of the $\theta_i$ are called the $f$-{\em sections} over $S$, and the regions between successive $f$-sections are called $f$-{\em sectors}.

\end{defn}

\begin{thm}\label{thm:McCallum}$\cite{McCallum1,McCallum2}$
	Let $f(\xx_n,x_{n+1})$ be a polynomial in $\ZZ[\xx_n,x_{n+1}]$ of positive degree and $\discrim(f,x_{n+1})$ is a nonzero polynomial. Let $S$ be a connected submanifold of $\RR^n$ on which $f$ is degree-invariant and does not vanish identically, and in which $\discrim(f,x_{n+1})$ is order-invariant. Then $f$ is analytic delineable on $S$ and is order-invariant in each $f$-section over $S$.
\end{thm}
Based on this theorem, McCallum proposed the projection operator \MCproj, which consists of the discriminant of $f$ and all coefficients of $f$.

\begin{thm}\label{thm:Brown}$\cite{brown}$
	Let $f(\xx_n,x_{n+1})$ be a $(n+1)$-variate polynomial of positive degree 
in $x_{n+1}$ such that $\discrim(f,x_{n+1})$ $\neq0$. Let $S$ be a connected submanifold of $\RR^n$ in which $\discrim(f,x_{n+1})$ is order-invariant, the leading coefficient of $f$ is sign-invariant, and such that $f$ vanishes identically at no point in $S$. $f$ is degree-invariant on $S$.
\end{thm}
Based on this theorem, Brown obtained a reduced McCallum projection in which only leading coefficients, discriminants and resultants appear. The Brown projection operator is defined as follows.
\begin{defn} \label{def:brown projection}$\cite{brown}$
	Given a polynomial $f\in \ZZ[\xx_n]$, if $f$ is with level $n$,
	the Brown projection operator for $f$ is
	$$\Bproj(f,[x_{n}])=\Res(\sqrfree(f),\frac{\partial (\sqrfree(f))}{\partial x_{n}}, x_{n}).$$
	Otherwise $\Bproj(f,[x_{n}])=f$.
	If $L$ is a polynomial set with level $n$, then
	\begin{align*}
		\Bproj(L,[x_{n}])=&\bigcup_{f\in L}\{\Res(\sqrfree(f),\frac{\partial (\sqrfree(f))}{\partial x_{n}}, x_{n})\}\\
		&\bigcup_{f,g\in L, f\neq g}\{\Res(\sqrfree(f),\sqrfree(g),x_{n})\}.
	\end{align*}
	Define
	\begin{align*}
		&\Bproj(f,[x_{n},x_{n-1},\ldots, x_i])\\
		=&\Bproj(\Bproj(f,[x_{n},x_{n-1},\ldots,x_{i+1}]),[x_i]).
	\end{align*}
\end{defn}

The following definition of {\em open CAD} is essentially the GCAD introduced in \cite{Strzebonski}. For convenience, we use the terminology of open CAD in this paper.

\begin{defn} $($Open CAD$)$ \label{def:opencad}
	For a polynomial $f(\xx_n)\in \ZZ[\xx_n]$, an open CAD defined by $f(\xx_n)$ is a set of sample points in $\RR^n$ obtained through the following three phases:\\
	(1) Projection. Use the Brown projection operator on $f(\xx_n)$, let $F=\{f, \Bproj(f,[x_n]),\dots, \Bproj(f,[x_n,\dots,x_2])\}$;\\ 
	(2) Base. Choose one rational point in each of the open intervals defined by the real roots of $F^{[1]}$;\\ 
	(3) Lifting. Substitute each sample point of $\RR^{i-1}$ for $\xx_{i-1}$ in $F^{[i]}$ to get a univariate polynomial $F_{i}(x_i)$ and then, by the same method as Base phase, choose sample points for $F_{i}(x_i)$. Repeat the process for $i$ from $2$ to $n$.
\end{defn}

\section{Reduced Open CAD}\label{sec:refined}

\begin{defn} $($Open sample$)$
A set of sample points $T_f\subseteq \RR^k$ is said to be an {\em open sample} defined by $f(\xx_k)\in\ZZ[\xx_k]$ in $\RR^k$ if it has the following property: for every open connected set $U\subseteq \RR^k$ defined by $f\neq0$, $T_f\cap U\ne \emptyset$.

Suppose $g(\xx_k)$ is another polynomial. If $T_f$ is an open sample defined by $f(\xx_k)$ in $\RR^k$ such that $g(\va)\neq0$ for any $\va\in T_f$, then we denote the open sample by $T_{f,g\ne0}$. 
\end{defn}

As a corollary of Theorems \ref{thm:McCallum} and \ref{thm:Brown}, a property of open CAD (or GCAD) is that at least one sample point can be taken from every highest dimensional cell via the open CAD (or GCAD) lifting phase.
So, an open CAD is indeed an open sample.

Obviously, there are various ways to compute $T_{f,g\ne0}$ for two given univariate polynomials $f,g \in \ZZ[x]$. Therefore, we only describe the specification of such algorithms here and omit the details of the algorithms.

\begin{algorithm}
	\caption{{\tt SPOne}} \label{one-dimsample}
	\begin{algorithmic}[1]
		\REQUIRE{Two univariate polynomials $f,g \in \ZZ[x]$}
		\ENSURE{$T_{f,g\ne0}$, an open sample defined by $f(x)$ in $\RR$ such that $g(\va)\neq0$ for any $\va\in T_{f,g\ne0}$}
       	\end{algorithmic}
\end{algorithm}

\begin{algorithm}
	\caption{{\tt OpenSP}} \label{opensample}
	\begin{algorithmic}[1]
		\REQUIRE{Two lists of polynomials $L_1=[f_n(\xx_n), \ldots, f_j(\xx_j)]$, $L_2=[g_n(\xx_n), \ldots, g_j(\xx_j)]$, and a set of points $T$ in $\RR^{j}$}
		\ENSURE{A set of sample points in $\RR^{n}$}
        \STATE $O:=T$
        \FOR {$i$ from $j+1$ to $n$}
        \STATE $P:=\emptyset$
        \FOR {$\va$ in $O$}
        \STATE $P:=P\bigcup (\va,{\tt SPOne}(f_i(\va,x_i),g_i(\va,x_i)))$
        \ENDFOR
        \STATE $O:=P$
        \ENDFOR
        \RETURN $O$
	\end{algorithmic}
\end{algorithm}

\begin{rem}
The output of ${\tt OpenSP}(L_1,L_2,T)$ is dependent on the method of choosing sample points in Algorithm $\tt SPOne$. In the following, when we use the terminology ``any ${\tt OpenSP}(L_1,L_2,T)$'', we mean ``no matter which method is used in Algorithm $\tt SPOne$ for choosing sample points".
\end{rem}

\begin{rem}
For a polynomial $f(\xx_n)\in \ZZ[\xx_n]$,
let
$B_1=[f, \Bproj(f,[x_n]),\dots, \Bproj(f,[x_n,\dots,x_2])]$, $B_2=[1,\ldots,1],$ and $T={\tt SPOne}(\Bproj(f,[x_n,\dots,x_2]),1)$,
then
${\tt OpenSP}(B_1,B_2,T)$
is an open CAD (an open sample) defined by $f(\xx_n)$.

We will provide in this section a method which computes two lists $C_1$ and $C_2$ where the polynomials in $C_1$ are factors of corresponding polynomials in $B_1$ and will prove that any ${\tt OpenSP}(C_1,C_2,T_{f_j,g_j\ne 0})$ is an open sample of $\RR^n$ defined by $f(\xx_n)$ for any open sample $T_{f_j, g_j\ne 0}$ in $\RR^j$ where $f_j\in C_1$ and $g_j\in C_2$.
\end{rem}

\begin{defn}\label{def:weakopendeli} $($Weak open delineable$)$
Let $S'$ be an open set of $\RR^j$. The polynomial $f_n(\xx_n)$ is said to be {\em weak open delineable} on $S'$ if, for any maximal open connected set $U\subset \RR^n$ defined by $f_n(\xx_n)\neq0$, we have
$(S'\times \RR^{n-j})\bigcap U\ne \emptyset \Longleftrightarrow \forall\va\in S', (\va\times \RR^{n-j})\bigcap U\ne \emptyset.$
\end{defn}

\begin{rem}
Let $S\subset \RR^{n-1}$ be a maximal open connected set and suppose $f_n$ is analytic delineable on $S$. It is clear that $f_n$ is weak open delineable on $S$. In some sense, analytic delineablility is stronger than weak open delineability. For example, let $f(x,y)=(x^2+y^2-1)(x^2+y^2)$, then $f$ is obviously weak open delineable on $(-1,1)$ of $x$-axis by Definition \ref{def:weakopendeli}, but $f$ is not analytic delineable on $(-1,1)$ because $f(0,y)$ has three different real roots, while $f(x,y)=0$ has only two different real roots when $x\neq0$ and $x\in (-1,1)$.
\end{rem}

\begin{defn}\label{def:opendeli} $($Open delineable$)$
Let
	\begin{align}
        & L_1=[f_n(\xx_n),f_{n-1}(\xx_{n-1}), \ldots,f_j(\xx_j)], \label{def:L1} \\
		& L_2=[g_n(\xx_n), g_{n-1}(\xx_{n-1}), \ldots,g_j(\xx_j)] \label{def:L2}
	\end{align}
be two polynomial lists, $S$ an open set of $\RR^s$ $(s\le j)$ and $S'=S\times \RR^{j-s}$. The polynomial $f_n(\xx_n)$ is said to be {\em open delineable} on $S$ w.r.t. $L_1$ and $L_2$, if $\mathcal{A}\bigcap U\ne \emptyset$ for any maximal open connected set $U\subset \RR^n$ defined by $f_n\neq0$ with $U\bigcap (S'\times \RR^{n-j})\ne \emptyset$ and any $\mathcal{A}={\tt OpenSP}(L_1,L_2,\{\va\})$ where $\va\in S'$ is any point such that $f_j(\va)g_j(\va)\ne 0$.
\end{defn}

\begin{rem}
Let $s=j$ in Definition \ref{def:opendeli}, it could be shown that if $f_n(\xx_n)$ is open delineable on $S'$ w.r.t. $L_1$ and $L_2$, then $f_n(\xx_n)$ is weak open delineable on $S'\backslash\zero\{f_jg_j\}$.

Suppose $f_n(\xx_n)$ is a squarefree polynomial in $\ZZ[\xx_n]$ of positive degree and $S\subset \RR^{n-1}$ is an open connected set in which $\Bproj(f_n,[x_{n}])$ is sign-invariant. According to Theorem \ref{thm:McCallum} and Theorem \ref{thm:Brown}, $f_n$ is analytic delineable on $S$. It is easy to see that $f_n$ is open delineable on $S$ w.r.t. $[f_n,\Bproj(f_n,[x_{n}])]$ and $[f_n,\Bproj(f_n,[x_{n}])]$.
\end{rem}

Open delineability has the following four properties. 
\begin{prop} $($open sample property$)$   \label{prop:sample}
Let $L_1, L_2$ be as in Definition \ref{def:opendeli}.
	If $f_n(\xx_n)$ is open delineable on every open connected set of $f_j(\xx_j)\neq0$ w.r.t. $L_1$ and $L_2$, then for any open sample $T_{f_j,g_j\ne 0}$ in $\RR^j$, any $\mathcal{A}={\tt OpenSP}(L_1,L_2,T_{f_j,g_j\ne 0})$ is an open sample defined by $f_n(\xx_n)$ in $\RR^n.$
\end{prop}
\begin{proof}
For any open connected set $U\subset \RR^n$ defined by $f_n\neq0$, there exists at least one open connected set $S\subset \RR^j$ defined by $f_j\neq0$ such that $U\bigcap (S\times \RR^{n-j})\ne \emptyset$. Since $f_n$ is open delineable on $S$ w.r.t. $L_1$ and $L_2$, we have $\mathcal{A}\bigcap U\ne \emptyset$ for any $\mathcal{A}={\tt OpenSP}(L_1,L_2,T_{f_j,g_j\ne 0})$.
\end{proof}

\begin{prop} $($transitive property$)$ \label{prop:transitive}
Let $L_1, L_2, S, S'$ be as in Definition \ref{def:opendeli}. Suppose that there exists $k (j\le k\le n)$ such that $f_k(\xx_k)$ is open delineable on $S$ w.r.t. $L_1''=[f_k(\xx_k), \ldots, f_j(\xx_j)]$ and $L_2''=[g_k(\xx_k), \ldots, g_j(\xx_j)]$, and $f_n(\xx_n)$ is open delineable on every open connected set of $f_k(\xx_k)\neq0$ w.r.t. $L_1'=[f_n(\xx_n), \ldots, f_k(\xx_k)]$ and $L_2'=[g_n(\xx_n), \ldots, g_k(\xx_k)]$. Then $f_n(\xx_n)$ is open delineable on $S$ w.r.t. $L_1$ and $L_2$.
\end{prop}
\begin{proof}
Let $\va\in S'$ be any point such that $f_j(\va)g_j(\va)\ne 0$, for any $\mathcal{A}={\tt OpenSP}(L_1,L_2,\{\va\})$, we have $\mathcal{A}={\tt OpenSP}(L_1',$ $L_2',\mathcal{A'})$ where $\mathcal{A'}={\tt OpenSP}(L_1'',L_2'',\{\va\})$. For any open connected set $U\subset \RR^n$ defined by $f_n\neq0$ with $U\bigcap (S'\times \RR^{n-j})\neq \emptyset$, there exists an open connected set $V\subseteq \RR^k$ defined by $f_k\neq0$ with $U\bigcap (V\times \RR^{n-k})\neq \emptyset$ and $V\bigcap (S\times \RR^{k-s})\neq \emptyset$. Now we have $\mathcal{A'}\bigcap V\neq\emptyset$ since $f_k(\xx_k)$ is open delineable on $S$ w.r.t. $L_1''$ and $L_2''$. And then, $\mathcal{A}\bigcap U\neq \emptyset$ is implied by $U\bigcap (V\times \RR^{n-k})\neq \emptyset$ since $f_n(\xx_n)$ is open delineable on $V$ w.r.t. $L_1'$ and $L_2'$.
\end{proof}

\begin{prop} $($nonempty intersection property$)$  \label{prop:nonempty}
Let $L_1$, $L_2$ be as in Definition \ref{def:opendeli}. For two open sets $S_1$ and $S_2$ of $\RR^s$ $(s\le j)$ with $S_1\bigcap S_2\neq \emptyset$, if $f_n(\xx_n)$ is open delineable on both $S_1$ and $S_2$ w.r.t. $L_1$ and $L_2$, then $f_n(\xx_n)$ is open delineable on $S_1\bigcup S_2$ w.r.t. $L_1$ and $L_2$.
\end{prop}

\begin{proof}
 For any $\va_1\in S_1$, $\va_2\in S_2$, $\va_3\in S_1\bigcap S_2$ with $f_j(\va_i)g_j(\va_i)\ne 0$, any $\mathcal{A}_i={\tt OpenSP}(L_1,L_2,\{\va_i\})$, and open connected set $U\subset \RR^n$ defined by $f_n\neq0$, we have $U\bigcap (S_1\times \RR^{n-s})\neq\emptyset \Longleftrightarrow \mathcal{A}_1\bigcap U\neq\emptyset \Longleftrightarrow \mathcal{A}_3\bigcap U\neq\emptyset\Longleftrightarrow \mathcal{A}_2\bigcap U\neq\emptyset \Longleftrightarrow U\bigcap (S_2 \times \RR^{n-s})\neq\emptyset$.
\end{proof}

\begin{prop} $($union property$)$ \label{prop:union}
Let $L_1, L_2$ be as in Definition \ref{def:opendeli}. For $\sigma\in P_{n,j+1}$, denote $\yy_n=(y_1,\ldots,y_n)=\sigma(\xx_n)$ and $\yy_i=(y_1,\ldots,y_i)$. Let $L'_1=[f_n(\xx_n)$, $p_{n-1}(\yy_{n-1})$, $\ldots,$ $p_j(\yy_j)]$ and $L'_2=[q_n(\xx_n)$, $q_{n-1}(\yy_{n-1})$, $\ldots,$ $q_j(\yy_j)]$ where $p_{i}(\yy_{i}) $ and $q_{i}(\yy_{i})$ are polynomials in $i$ variables.\par
For two open sets $S_1$ and $S_2$ of $\RR^j$, if
(a) $f_n(\xx_n)$ is open delineable on both $S_1$ and $S_2$ w.r.t. $L_1$ and $L_2$, (b) $f_n(\xx_n)$ is open delineable on $S_1\bigcup S_2$ w.r.t. $L'_1$ and $L'_2$, and (c) $p_j(\yy_j)q_j(\yy_j)$ vanishes at
no points in $S_1\bigcup S_2$, then $f_n(\xx_n)$ is open delineable on $S_1\bigcup S_2$ w.r.t. $L_1$ and $L_2$.
\end{prop}
\begin{proof}
		Let $\va_1\in S_1,\va_2\in S_2$ be two points such that $g_jp_jq_j(\va_t)\neq0$ for $t=1,2$. Let $\mathcal{A}_t={\tt OpenSP}(L_1,L_2,\{\va_t\})$ and $\mathcal{A'}_t={\tt OpenSP}(L_1',L_2',\{\va_t\})$.	
		For any open connected set $U$ defined by $f_n\neq0$ with $U\bigcap (\va_1\times \RR^{n-j})\neq\emptyset$, then $\mathcal{A}_1\bigcap U\neq\emptyset$ and $\mathcal{A'}_1\bigcap U\neq\emptyset$. Since $f_n(\xx_n)$ is open delineable on $S_1\bigcup S_2$ w.r.t. $L'_1$ and $L'_2$, we have $\mathcal{A'}_2\bigcap U\neq \emptyset$ which implies that $U\bigcap (S_2\times \RR^{n-j})\neq\emptyset$ and $\mathcal{A}_2\bigcap U\neq\emptyset$. Therefore, $f_n(\xx_n)$ is open delineable on $(S_1\bigcup S_2)\backslash \zero(p_jq_j)$ w.r.t. $L_1$ and $L_2$. Since $p_jq_j$ does not vanish at any point of $S_1\bigcup S_2$, $f_n(\xx_n)$ is open delineable on $S_1\bigcup S_2$ w.r.t. $L_1$ and $L_2$.	
\end{proof}

\begin{rem} \label{rem:union}
We use Example 1 to illustrate Proposition \ref{prop:union}. Let $L_1=[f,f_{z},f_{zy}]$, $L_1'=[f,f_{y},f_{yz}]$, $L_2=L_1,$ $L_2'=L_1'$, where $f_y,f_{yz},f_z,f_{zy}$ are defined in Example 1. 
Let $S_1=(-1,1)$, $S_2=(1,\sqrt{\sqrt{5}-1})$ be two open intervals in $x$-axis, where $x=1$ is one of the real roots of the equation $f_{zy}=0$. By typical CAD methods, $S_1$ and $S_2$ are two different cells in $x$-axis.

It could be deduced easily by Theorem \ref{thm:McCallum}, Theorem \ref{thm:Brown}, and Proposition \ref{prop:transitive} that the conditions (a) and (b) of Proposition \ref{prop:union} are satisfied.
Since $f_{yz}$ vanishes at no no points in $S_1\bigcup S_2$, condition (c) is also satisfied.

By Proposition \ref{prop:union}, $f$ is open delineable on $S_1\bigcup S_2$ w.r.t. $L_1$ and $L_2$. Roughly speaking, the real root of $x-1$ would not affect the open delineability, thus we could combine the two cells $S_1$ and $S_2$. 

For the same reason, the real root of $x+1$ would not affect the open delineability either. 
\end{rem}

Now, we define the new projection operator $\Hproj$.
\begin{defn}\label{def:hp}
Let $f\in \ZZ[x_1,\dots,x_n]$. 
For $m (1\le m\le n),$ denote $[\yy]=[y_1,\dots,y_{m}]$ where $ y_i \in \{x_1,\dots,x_n\}$ for $1\le i\le m$ and $y_i\neq y_j$ for $i\neq j$.  For $1\le i\le m$,
$\Hproj(f,[\yy],y_i)$ and $\Hproj(f,[\yy])$ are defined recursively as follows.
	\begin{align*}
		&\Hproj(f,[\yy],y_i)=\Bproj(\Hproj(f,[\hat{\yy]}_i),[y_i]),\\
		&\Hproj(f,[\yy])=\gcd(\Hproj(f,[\yy],y_1),\ldots,\Hproj(f,[\yy],y_m)),
    \end{align*}
where $\hat{[\yy]}_i=[y_1,\ldots,y_{i-1},y_{i+1},\ldots,y_{m}]$ and $\Hproj(f,[~])=f$. Define
\[\overline{\Hproj}(f,i)=\{f,\Hproj(f,[x_{n}]), \ldots, \Hproj(f,[x_{n},\ldots,x_i])\},\] and
\[\widetilde{\Hproj}(f,i)=\{f,\Hproj(f,[x_n],x_n),  \ldots, \Hproj(f,[x_n,\ldots,x_{i}],x_i)\}.\]
\end{defn}

A {\em reduced open CAD} of $f(\xx_n)$ w.r.t. $[x_n,\ldots,x_j]$ is a set of sample points in $\RR^n$ obtained through Algorithm \ref{de:reducedopencad}.
\begin{algorithm}
	\caption{{\tt ReducedOpenCAD}}  \label{de:reducedopencad}
	\begin{algorithmic}[1]
		\REQUIRE{A polynomial $f(\xx_n)\in\ZZ[\xx_n]$, and an open sample $T_{\Hproj(f,[x_n,\ldots,x_{j+1}]),\Hproj(f,[x_n,\dots,x_{j+1}], x_{j+1} )\ne0}$ in $\RR^{j}$} 
		\ENSURE{A set of sample points in $\RR^{n}$}
        \STATE $O:=T_{\Hproj(f,[x_n,\ldots,x_{j+1}]),\Hproj(f,[x_n,\dots,x_{j+1}], x_{j+1} )\ne0}$
        \FOR {$i$ from $j+2$ to $n+1$}
        \STATE $P:=\emptyset$
        \FOR {$\va$ in $O$}
        \IF {$i\le n$}
        \STATE $P:=P\bigcup (\va, {\tt SPOne}(\Hproj(f,[x_n,\ldots,x_{i}])(\va,x_{i-1}),$\\
         $\Hproj(f,[x_n,\dots,x_{i}],x_{i})(\va,x_{i-1})))$
        \ELSE \STATE $P:=P\bigcup (\va, {\tt SPOne}(f(\va,x_{n}),f(\va,x_{n})))$
        \ENDIF
        \ENDFOR
        \STATE $O:=P$
        \ENDFOR
        \RETURN $O$
	\end{algorithmic}
\end{algorithm}

\begin{lem}\label{thm:1} $\cite{han2012}$ Let $f$ and $g$ be coprime in $\ZZ[\xx_n]$. For any connected open set $U$ of $\RR^n$, the open set $V = U\backslash \zero(f,g)$ is also connected.
\end{lem}

\begin{lem}\label{lem:connected} Let $f=\gcd(f_1,\ldots,f_m)$ where $f_i\in\ZZ[\xx_n]$, $i=1,2,\ldots,m.$ Suppose $f$ has no real roots in a connected open set $U\subseteq\RR^n$, then the open set $V = U\backslash \zero(f_1,\ldots,f_m)$ is also connected.
\end{lem}
\begin{proof}
Without loss of generality, we can assume that $f=1$. If $m=1,$ the result is obvious. The result of case $m=2$ is just the claim of Lemma \ref{thm:1}. For $m\ge3$, let $g=\gcd(f_1,\ldots,f_{m-1})$ and $g_i=f_i/g$ $(i=1,\ldots,m-1)$, then $\gcd(f_m,g)=1$ and $\gcd(g_1,\ldots,g_{m-1})=1$. Let $A=\zero(f_1,\ldots,f_m)$, $B=\zero(g_1,\ldots,g_{m-1})\bigcup \zero(g,f_m)$. Since $A\subseteq B$, we have $U\backslash B\subseteq U\backslash A$. Notice that the closure of ${U\backslash B}$ equals the closure of ${U\backslash A}$, it suffices to prove that $U\backslash B$ is connected, which follows directly from Lemma \ref{thm:1} and induction.
\end{proof}

As a Corollary of Theorem \ref{thm:McCallum} and Theorem \ref{thm:Brown}, we have
\begin{prop}\label{prop:open}
	Let $f\in\ZZ[\xx_{n}]$ be a squarefree polynomial with level $n$. Then $f$ is open delineable on every open connected set defined by $\Bproj(f,[x_{n}])\neq0$ in $\RR^{n-1}$ w.r.t. $\overline{\Hproj}(f,n)$ and $\widetilde{\Hproj}(f,n)$.
\end{prop}

The following Theorem is the main result of this paper, which shows that the reduced open CAD owns the property of open delineability.
\begin{thm} \label{thm:open delineable}
Let $j$ be an integer and $2\le j\le n$. For any given polynomial $f(\xx_n)\in \ZZ[\xx_n]$ and any open connected set $U\subset \RR^{j-1}$ of $\Hproj(f,[x_n,\ldots,x_j])\neq0$, let $S=U\backslash\zero(\{\Hproj(f,[x_n,\dots,x_j],x_t)\mid t=j,\ldots,n\})$. Then $f(\xx_n)$ is open delineable on the open connected set $S$ w.r.t. $\overline{\Hproj}(f,j)$ and $\widetilde{\Hproj}(f,j)$. As a result, a reduced open CAD of $f(\xx_n)$ w.r.t. $[x_n,\ldots,x_j]$ is an open sample defined by $f(\xx_n)$. 
\end{thm}
\begin{proof} First, by Lemma \ref{lem:connected}, $S$ is open connected.
We prove the theorem by induction on $k=n-j$. When $k=0$, it is obvious true from Proposition \ref{prop:open}. Suppose the theorem is true for all polynomials $g(\xx_k)\in \ZZ[\xx_k]$ with $k=0,1,\ldots,n-i-1$. We now consider the case $k=n-i$. Let $[\zz]=[x_n,\ldots,x_i]$. For any given polynomial $f(\xx_n)\in \ZZ[\xx_n]$, let $U\subset \RR^{i-1}$ be an open connected set of $\Hproj(f,[\zz])\neq0$ and $S=U\backslash\zero(\{\Hproj(f,[\zz],x_t)\mid t=i,\ldots,n\})$.

For any point $\va\in S$ with $\Hproj(f,[\zz] ,{x_i})(\va)\ne0$, there exists an open connected set $S_{\va}\subset\RR^{i-1}$ such that $\va\in S_{\va}$ and $0\notin\Hproj(f,[\zz],{x_i} )(S_{\va})$. By induction, $\Hproj(f,[x_n,\ldots,x_{i+1}])$ is open delineable on $S_{\va}$ w.r.t. $\{\Hproj(f,[\zz])\}$ and $\{\Hproj(f,[\zz],{x_i})\}$. By induction again and the transitive property of open delineable $($Proposition \ref{prop:transitive}$)$, $f$ is open delineable on $S_{\va}$ w.r.t. $\overline{\Hproj}(f,i)$ and $\widetilde{\Hproj}(f,i)$.

For any point $\va\in S$ with $\Hproj(f,[\zz],{x_i} )(\va)=0$, there exists an $i'$ such that $n\ge i'\ge i+1$ and $\Hproj(f,[\zz],{x_{i'}} )(\va)\ne0$. Thus there exists an open connected set $S'_{\va}$ of $\RR^{i-1}$ such that $\va\in S'_{\va}$ and
		$0\notin\Hproj(f,[\zz],x_{i'} )(S'_{\va})$. Let $\sigma\in P_{n,i}$ with $\sigma(x_i)=x_{i'}$, in such case, $f(\sigma(\xx_n))$ is open delineable on $S'_{\va}$ w.r.t. $\overline{\Hproj}(f(\sigma(\xx_n)),i)$ and $\widetilde{\Hproj}(f(\sigma(\xx_n)),i)$. For any $\vb\in S'_{\va}$ with $\Hproj(f,[\zz],x_i)(\vb)\ne 0$, there exists an open connected set $S''_{\va}\subset S'_{\va}$ and $f$ is open delineable on $S''_{\va}$ w.r.t. $\overline{\Hproj}(f,i)$ and $\widetilde{\Hproj}(f,i)$.
		From union property of open delineable (Proposition \ref{prop:union}), $f$ is open delineable on $S'_{\va}$ w.r.t. $\overline{\Hproj}(f,i)$ and $\widetilde{\Hproj}(f,i)$.\par
		To summarize, the above discussion shows that for any point $\va\in S$, there exists an open connected set $S_{\va}\subset S$ such that $\va\in S_{\va}$ and $f$ is
		open delineable on $S_{\va}$ w.r.t. $\overline{\Hproj}(f,i)$ and $\widetilde{\Hproj}(f,i)$. By the nonempty intersection property of open delineable $($Proposition \ref{prop:nonempty}$)$ and the fact that $S$ is connected, $f(\xx_n)$ is open delineable on $S$ w.r.t. $\overline{\Hproj}(f,i)$ and $\widetilde{\Hproj}(f,i)$ as desired.

Therefore, the theorem is proved by induction. The last statement of the theorem follows from Proposition \ref{prop:sample}. 
	\end{proof}

\begin{rem}\label{re:a1}
As an application of Theorem \ref{thm:open delineable}, we could design a CAD-like method to get an open sample defined by $f(\xx_n)$ for a given polynomial $f(\xx_n)$. Roughly speaking, if we have already got an open sample defined by $\Hproj(f,[x_n,\ldots,x_j])$ in $\RR^{j-1}$, according to Theorem \ref{thm:open delineable}, we could obtain an open sample defined by $f$. That process could be done recursively.

In the definition of $\Hproj$, we first choose $m$ variables from $\{x_1,...,x_n\}$, compute all projection polynomials under all possible orders of those $m$ variables, and then compute the gcd of all those projection polynomials.
Therefore, Theorem \ref{thm:open delineable} provides us many ways for designing various algorithms for computing open samples. For example, we may set $m=2$ and choose $[x_n,x_{n-1}]$, $[x_{n-2}, x_{n-3}]$, etc. successively in each step. Because there are only two different orders for two variables, we compute the gcd of two projection polynomials under the two orders in each step. Algorithm \ref{TwoHp} is based on this choice.
\end{rem}

\begin{algorithm}
	\caption{\TwoHp} \label{TwoHp}
	\begin{algorithmic}[1]
		\REQUIRE{A polynomial $f \in \ZZ[\xx_n]$ of level $n$.}
		\ENSURE{An open sample defined by $f$, i.e., a set of sample points which contains at least one point from each connected component of $f\neq0$ in $\RR^{n}$}
        \STATE $g:=f$;
        \STATE $L_1:=\{\}$;
        \STATE $L_2:=\{\}$;
\WHILE{$i\ge 3$}
        \STATE $L_1:=L_1\bigcup \overline{\Hproj}(g,i-1)$;
		\STATE $L_2:=L_2\bigcup \widetilde{\Hproj}(g,i-1)$;
        \STATE $g:=\Hproj(g,[x_{i},x_{i-1}])$;
         \STATE $i:=i-2$;
\ENDWHILE
\IF {$i=2$}
        \STATE $L_1:=L_1\bigcup \overline{\Hproj}(g,i)$;
		\STATE $L_2:=L_2\bigcup \widetilde{\Hproj}(g,i)$;
        \STATE $g:=\Hproj(g,[x_{i}])$;
        \ENDIF
        \STATE $T$:=${\tt SPOne}(L_1^{[1]},L_2^{[1]})$;
        \STATE $C$:= ${\tt OpenSP}(L_1,L_2,T)$;
        \RETURN $C$.
	\end{algorithmic}
\end{algorithm}
\begin{rem}\label{re:a2}
If $\Hproj(f,[x_{n},x_{n-1})]\neq \Bproj(f,[x_{n},x_{n-1}])$ and $n>3$, it is obvious that the scale of projection in Algorithm \ref{TwoHp} is smaller than that of open CAD in Definition \ref{def:opencad}.
\end{rem}
\begin{rem}
It could be shown that if we modify the definition of $\Hproj$ by choosing several (not all) orders of those $m$ variables and computing the gcd of the projection polynomials under those orders, Theorem \ref{thm:open delineable} is still valid. Due to page limit, we will give a proof of this claim in our future work.
\end{rem}

\section{
projection operator Np}
\label{sec:improved}
In this section, we combined the idea of $\Hproj$ and the simplified CAD projection operator $\Nproj$ we introduced previously in \cite{han2012}, to get a new algorithm for testing semi-definiteness of polynomials.
\begin{defn} $\cite{han2012}$
	Suppose $f\in \ZZ[\xx_{n}]$ is a polynomial of level $n$. Define
	\begin{align*}
		& {\rm Oc}(f,x_n)=\sqrfree_1(\lc(f,x_{n})), {\rm Od}(f,x_n)=\sqrfree_1(\discrim(f,x_{n})),\\
		& {\rm Ec}(f,x_n)=\sqrfree_2(\lc(f,x_{n})), {\rm Ed}(f,x_n)=\sqrfree_2(\discrim(f,x_{n})),\\
		& {\rm Ocd}(f,x_n)={\rm Oc}(f,x_n)\cup {\rm Od}(f,x_n),\\
		& {\rm Ecd}(f,x_n)={\rm Ec}(f,x_n)\cup {\rm Ed}(f,x_n).
	\end{align*}
	The {\em secondary} and {\em principal parts} of the projection operator $\Nproj$ are defined as
	\begin{align*}
		\Nproj_1(f,[x_{n}])=&{\rm Ocd}(f,x_n),\\
		\Nproj_{2}(f,[x_{n}])=&\{\prod_{g\in {\rm Ecd}(f,x_{n})\setminus {\rm Ocd}(f,x_{n})}{g}\}.
	\end{align*}
	If $L$ is a set of polynomials of level $n$, define
	\begin{align*}
		\Nproj_1(L,[x_{n}])&=\bigcup_{g\in L}{\rm Ocd}(g,x_{n}),\\
		\Nproj_{2}(L,[x_{n}])&=\bigcup_{g\in L}\{\prod_{h\in {\rm Ecd}(g,x_n)\setminus \Nproj_1(L,[x_{n}])}{h}\}.
	\end{align*}
\end{defn}
Based on the projection operator $\Nproj$, we proposed an algorithm, \Proineq, in \cite{han2012} for proving polynomial inequalities. Algorithm \Proineq\ takes a polynomial $f(\xx_n) \in \ZZ[\xx_n]$ as input, and returns whether or not $f(\xx_n) \ge0$ on $\RR^n$.
The readers are referred to \cite{han2012} for the details of \Proineq.

The projection operator $\Nproj$ is extended and defined in the next definition.
\begin{defn}
	Let $f\in \ZZ[x_1,\dots,x_n]$ with level $n$. Denote $[\yy]=[y_1,\dots,y_{m}]$, for $1\le m\le n$, where $ y_i \in \{x_1,\dots,x_n\}$ for $1\le i\le m$ and $y_i\neq y_j$ for $i\neq j$. Define
 \[\Nproj(f,[x_i])=\Nproj_2(f,[x_i]), \Nproj(f,[x_i],x_i)=\prod_{g\in\Nproj_1(f,[x_i])}{g}.\]
 For $m(m\ge2)$ and $i(1\le i \le m)$, $\Nproj(f,[\yy],y_i)$ and $\Nproj(f,[\yy])$ are defined recursively as follows.
	\begin{align*}
&\Nproj(f,[\yy],y_i)=\Bproj(\Nproj(f,\hat{[\yy]}_i),y_i), \\
&\Nproj(f,[\yy])=\gcd(\Nproj(f,[\yy],y_1),\ldots,\Nproj(f,[\yy],y_m)),
\end{align*}
where $\hat{[\yy]}_i=[y_1,\ldots,y_{i-1},y_{i+1},\ldots,y_m]$. Define
\[\overline{\Nproj}(f,i)=\{f,\Nproj(f,[x_{n}]), \ldots, \Nproj(f,[x_{n},\ldots,x_i)]\},\] and
\[\widetilde{\Nproj}(f,i)=\{f,\Nproj(f,[x_n],x_n),\ldots,\Nproj(f,[x_n,\dots,x_i],x_i)\}.\]
\end{defn}

\begin{thm}\label{thm:2} $\cite{han2012}$
	Given a positive integer $n\ge2$. Let $f\in \ZZ[\xx_n]$ be a non-zero squarefree polynomial and $U$ a connected component of $\Nproj(f,[x_{n}])\neq0$ in $\RR^{n-1}$. If the polynomials in $\Nproj_{1}(f,[x_n])$ are semi-definite on $U$, then $f$ is delineable on $V=U\backslash \bigcup_{h\in \Nproj_{1}(f,[x_n])}\zero(h)$.
\end{thm}
\begin{lem} \label{prop:han2012} $\cite{han2012}$
	Given a positive integer $n\ge2$. Let $f\in\ZZ[\xx_{n}]$ be a squarefree polynomial with level $n$ and $U$ a connected open set of $\Nproj(f,[x_{n}])\neq0$ in $\RR^{n-1}$. If $f(\xx_n)$ is semi-definite on $U\times \RR$, then the polynomials in $\Nproj_{1}(f,[x_n])$ are all semi-definite on $U$.
\end{lem}
Now, we can rewritten Theorem \ref{thm:2} in another way. 
\begin{prop}\label{pr:6}
	Let $f\in\ZZ[\xx_{n}]$ be a squarefree polynomial with level $n$ and $U$ a connected component of $\Nproj(f,[x_{n}])\neq0$ in $\RR^{n-1}$. If the polynomials in $\Nproj_{1}(f,[x_n])$ are semi-definite on $U$, then $f$ is open delineable on $U$ w.r.t. $\overline{\Nproj}(f,n)$ and $\widetilde{\Nproj}(f,n)$.
\end{prop}
Notice that the proof of Theorem \ref{thm:open delineable} only uses the properties of open delineable (Propositions 1-4) and Proposition \ref{prop:open}, and Proposition \ref{pr:6} is similar to Proposition \ref{prop:open}. We can prove the following theorem by the same way of proving Theorem \ref{thm:open delineable}.

\begin{thm} \label{thm:nprojopen}
	Let $j$ be an integer and $2\le j\le n$. For any given polynomial $f(\xx_n)\in \ZZ[\bm{x}_n]$, and any open connected set $U$ of $\Nproj(f,[x_n,\ldots,x_j])\neq0$ in $\RR^{j-1}$, let $S=U\backslash\zero(\{\Nproj(f,[x_n,\dots,x_j],x_t)\mid t=j,\ldots,n\})$. If the polynomials in $\bigcup_{i=0}^{n-j} \Nproj_{1}(f,[x_{n-i}])$ are all semi-definite on $U\times \RR^{n-j}$, $f(\xx_n)$ is open delineable on $S$ w.r.t. $\overline{\Nproj}(f,j)$ and $\widetilde{\Nproj}(f,j)$.
\end{thm}

Theorem \ref{thm:nprojopen} and Proposition \ref{prop:han2012} provide us a new way to decide the non-negativity of a polynomial as stated in the next theorem.

\begin{thm}\label{th:6}
	Given a positive integer $n$. Let $f\in\ZZ[\xx_{n}]$ be a squarefree polynomial with level $n$ and $U$ a connected open set of $\Nproj(f,[x_{n},\ldots,x_j])\neq0$ in $\RR^{j-1}$. Denote $S=U\backslash\zero(\{\Nproj(f,[x_n,\dots,x_j],x_t)\mid t=j,\ldots,n\})$. The necessary and sufficient condition for $f(\xx_n)$ to be positive semi-definite on $U\times \RR^{n-j+1}$ is the following two conditions hold.\\
	$(1)$The polynomials in $\bigcup_{i=0}^{n-j} \Nproj_{1}(f,[x_{n-i}])$ are all semi-definite on $U\times \RR^{n-j}$.\\
	$(2)$There exists a point $\va \in S$ such that $f(\va,x_j,\ldots,x_n)$ is positive semi-definite on $\RR^{n-j+1}$.
\end{thm}
Based on the above theorems, it is easy to design some different algorithms (depending on the choice of $j$) to prove polynomial inequality. For example, the algorithm \TwoPro\ for deciding whether a polynomial is positive semi-definite, which we will introduce later, is based on Theorem \ref{th:6} when $j=n-1$ (Proposition \ref{prop:nproj2}).
\begin{prop}\label{prop:nproj2}
	Given a positive integer $n\ge3$. Let $f\in\ZZ[\xx_{n}]$ be a squarefree polynomial with level $n$ and $U$ a connected open set of $\Nproj(f,[x_{n},x_{n-1}])\neq0$ in $\RR^{n-2}$. Denote $S=U\backslash \zero(\Nproj(f,[x_n,x_{n-1}],x_n),\Nproj(f,[x_n,x_{n-1}],x_{n-1}))$.\\ The necessary and sufficient condition for $f(\xx_n)$ to be positive semi-definite on $U\times \RR^2$ is the following two conditions hold.\\
	$(1)$The polynomials in either $\Nproj_{1}(f,[x_n])$ or $\Nproj_{1}(f,[x_{n-1}])$ are semi-definite on $U\times \RR$.\\
	$(2)$There exists a point $\va\in S$ such that $f(\va,x_{n-1},x_n)$ is positive semi-definite on $\RR^2$.
\end{prop}
\begin{algorithm}
	\caption{\TwoPro} \label{TwoPro}
	\begin{algorithmic}[1]
		\REQUIRE{An irreducible polynomial $f \in \ZZ[\xx_n]$.}
		\ENSURE{Whether or not $\forall \va_n\in \RR^n$, $f(\va_n)\ge0.$}
        \IF {$n\le2$}
        \IF {$\Proineq(f(x_n))$=\textbf{false}}
        \RETURN  \textbf{false}
        \ENDIF
        \ELSE
		\STATE $L_1:=\Nproj_{1}(f,[x_n])\bigcup \Nproj_{1}(f,[x_{n-1}])$
		\STATE $L_2:=\Nproj(f,[x_{n},x_{n-1}])$
		\FOR {$g$ in $L_1$}
		\IF {$\TwoPro(g)=$\textbf{false}}
		\RETURN  \textbf{false}
		\ENDIF
		\ENDFOR
		\STATE $C_{n-2}:=$ A reduced open CAD of $L_2$ w.r.t. \\
$[x_{n-2},\ldots,x_2]$, which satisfies that \\
$\zero(\Nproj(f,[x_n,x_{n-1}],x_n),\Nproj(f,[x_n,x_{n-1}],x_{n-1}))$\\
$\cap C_{n-2}=\emptyset$.
		\IF{$\exists \va_{n-2}\in C_{n-2}$ such that \\
$\Proineq(f(\va_{n-2},x_{n-1},x_n))$=\textbf{false}}
		\RETURN \textbf{false}
        \ENDIF
		\ENDIF\\
		\RETURN\textbf{true}
	\end{algorithmic}
\end{algorithm}

\section{Examples}
\label{sec:applicat}
The Algorithm \TwoHp\ and Algorithm \TwoPro\ have been implemented as two programs using Maple. In this section, we report the performance of the two programs, respectively.
All the timings in the tables are in seconds.

\begin{ex}\label{ex:61}
In this example, we compare the performance of Algorithm \TwoHp\ with open CAD on randomly generated polynomials. All the data in this example were obtained on a PC with Intel(R) Core(TM) i5 3.20GHz CPU, 8GB RAM, Windows 7 and Maple 17.

In the following table, we list the average time of projection phase and lifting phase, and the average number of sample points on $30$ random polynomials with $4$ variables and degree $4$ generated by {\tt randpoly([x,y,z,w],degree=4)-1}.

\begin{center}
		\begin{tabular}{ccccc}
			\hline & {\rm Projection} &{\rm Lifting}& {\rm Sample\ points}  & \\
			\hline
			$\TwoHp$           &$0.13$& $ 0.29$& $262$&\\
			${\tt open\ CAD}$         & $0.19$&$ 3.11$&$ 486$&\\
			\hline
		\end{tabular}
	\end{center}
If we get random polynomials with $5$ variables and degree $3$ by the command ${\tt randpoly([seq(x[i],i=1..5)], degree=3)}$, then the degrees of some variables are usually one. That makes the computation very easy for both \TwoHp\ and open CAD. Therefore,
we run the command ${\tt randpoly([seq(x[i],}$ ${\tt i=1..5)], degree=3)+add(x[i]^2,i=1..5)-1}$ 
 ten times to generate $10$ random polynomials with $5$ variables and degree $3$. 
The data on the $10$ polynomials are listed in the following table.
\begin{center}
		\begin{tabular}{ccccc}
			\hline & {\rm Projection} &{\rm Lifting}& {\rm Sample\ points}  & \\
			\hline
			$\TwoHp$           &$2.87$& $ 3.51$& $2894$&\\
			${\tt open\ CAD}$         & $0.76$&$ 12.01$&$7802$&\\
			\hline
		\end{tabular}
	\end{center}
For many random polynomials with $4$ variables and degree greater than $4$ (or $5$ variables and degree greater than $3$), neither \TwoHp\ nor open CAD can finish computation in reasonable time. 
\end{ex}

A main application of the new projection operator \Hproj\ is testing semi-definiteness of polynomials. 
Now, we illustrate the performance of our implementation of Algorithm \TwoPro\ with several non-trivial examples. For more examples, please visit the homepage\footnote{\url{https://sites.google.com/site/jingjunhan/home/software}} of the first author.

We report the timings of the program \TwoPro, the program \Proineq\ \cite{han2012}, the function PartialCylindricalAlgebraicDecomposition (\PCAD) in Maple 15, function FindInstance (\FI) in Mathematica 9, QEPCAD B (\QEPCAD), the program \RAGlib \footnote{\RAGlib\ release 3.19.4 (Oct., 2012).}, and {\tt SOSTOOLS} in MATLAB \footnote{The MATLAB version is R2011b, SOSTOOLS's version is 3.00 and SeDuMi's version is 1.3.} on these examples.

\QEPCAD\ and {\tt SOSTOOLS} were performed on a PC with Intel(R) Core(TM) i5 3.20GHz CPU, 4GB RAM and ubuntu.
The other computations were performed on a laptop with Inter Core(TM) i5-3317U 1.70GHz CPU, 4GB RAM, Windows 8 and Maple 15.

\begin{ex}\label{ex:62} $\cite{han2011}$ Prove that
	$$F(\bm{x}_{n})=(\sum_{i=1}^nx_i^2)^2-4\sum_{i=1}^n x_i^2x_{i+1}^2\ge 0,$$
	where $x_{n+1}=x_1$.

Hereafter ``$\infty$" means either the running time is over 4000 seconds or the software is failure to get an answer.
	\begin{center}
		\begin{tabular}{lllllll}
			\hline$n$ & 5 & 8& 11  & 17  &23  \\
			\hline
			$\TwoPro$        &0.28 &0.95&6.26&29.53&140.01 \\
			$\RAGlib$        &6.98 &177.75& $\infty$&$\infty$&$\infty$\\
			$\Proineq$       &0.29 &$\infty$&$\infty$&$\infty$&$\infty$\\
            $\FI$            &0.10 &$\infty$ &$\infty$&$\infty$&$\infty$\\
            $\PCAD$          &0.26 &$\infty$ &$\infty$&$\infty$&$\infty$\\
            $\QEPCAD$        &0.10 &$\infty$ &$\infty$&$\infty$&$\infty$\\
            {\tt SOSTOOLS}   &0.23 & 1.38    & 3.94  & 247.56  & $\infty$\\
			\hline
		\end{tabular}
	\end{center}

We then test the semi-definiteness of the polynomials (In fact, all $G(\bm{x}_{n})$ are indefinite.)
$$G(\bm{x}_{n})=F(\bm{x}_{n})-\frac{1}{10^{10}}x_1^4.$$
The timings are reported in the following table.

	\begin{center}
		\begin{tabular}{llllllll}
			\hline$n$ & \TwoPro & \RAGlib&   \Proineq &  \FI& \PCAD&\QEPCAD\\
			\hline
			$20$        &3.828&$\infty$& $\infty$& $\infty$&$\infty$&$\infty$\\
			$30$       &13.594&$\infty$& $\infty$& $\infty$&$\infty$&$\infty$\\
           \hline
		\end{tabular}
	\end{center}
\end{ex}

\begin{ex} Prove that
	$$B(\bm{x}_{3m+2})=(\sum_{i=1}^{3m+2}x_i^2)^2-2\sum_{i=1}^{3m+2}x_i^2\sum_{j=1}^mx_{i+3j+1}^2\ge 0,$$
	where $x_{3m+2+r}=x_r$.
	If $m=1$, it is equivalent to the case $n=5$ of Example \ref{ex:62}. This form was once studied in $\cite{parrilo2000structured}$.
\end{ex}
	\begin{center}
		\begin{tabular}{llllllll}
			\hline$m$ & \TwoPro& \RAGlib&   \Proineq   &  \FI& \PCAD&\QEPCAD                \\
			\hline
			$1$        &0.296&6.9&0.297&0.1&0.26&0.104 \\
			$2$        &1.390&144.9&23.094&$\infty$&$\infty$&$\infty$\\
			$3$        &9.672&2989.5&$\infty$&$\infty$&$\infty$&$\infty$\\
			\hline
		\end{tabular}
	\end{center}

\begin{rem}
For some special examples like Example \ref{ex:62}, \TwoPro\ could solve problems with more than 30 variables efficiently. Of course, there also exist some other examples on which \TwoPro\ performs badly. For example, \TwoPro\ could not solve the problems in \cite{kaltofen2009proof} within $4000$ seconds while they can be solved by \RAGlib\ efficiently.

As showed by Example \ref{ex:61}, according to our experiments, the application of \TwoHp\ and \TwoPro\ is limited at $3$-$4$ variables and low degrees generally.
It is not difficult to see that, if the input polynomial $f(\xx_n)$ is symmetric, the new projection operator \Hproj\ cannot reduce the projection scale and the number of sample points.
Thus, it is reasonable to conclude that the complexity of \TwoPro\ is still doubly exponential.
\end{rem}

\section{Conclusion}

In this paper, we propose a new projection operator $\Hproj$ based on Brown's operator and gcd computation. The new operator computes the intersection of projection factor sets produced by different CAD projection orders. In other words, it computes the gcd of projection polynomials in the same variables produced by different CAD projection orders. In some sense, the polynomial in the projection factor sets of $\Hproj$ is irrelevant to the projection orders. We prove that the new operator still guarantees obtaining at least one sample point from every connected component of the highest dimension, and therefore, can be used for testing semi-definiteness of polynomials. In many cases, the new operator produces smaller projection factor sets and thus fewer open cells. Some examples of testing semi-definiteness of polynomials, which are difficult to be solved by existing tools, have been worked out efficiently by our program \TwoPro\ based on the new operator.

On the other hand, the complexity of the new algorithm \TwoPro\ is still doubly exponential and thus, it cannot be expected that \TwoPro\ always works more efficient than typical CAD methods.

\section{Acknowledgements}
The work was supported by National Science Foundation of China Grants 11290141 and 11271034 and the SKLCS project SYSKF1207.

The authors would like to convey their gratitude to Hoon Hong who provided his valuable comments, advice and suggestion on this paper when he visited Peking University. Thanks also go to M. Safey El Din who provided us several examples and communicated with us on the usage of \RAGlib.

The authors would like to convey their gratitude to all the four referees who provided their valuable comments, advice and suggestion, which help improve this paper greatly on not only the presentation but also the technical details.

\bibliographystyle{abbrv}

\end{document}